\documentclass[journal,onecolumn,romanappendices]{IEEEtran}

\usepackage{amsmath}
\usepackage{mathtools}
\usepackage{amsfonts}
\usepackage{csquotes}
\usepackage{amssymb}
\usepackage{amsthm}
\usepackage{enumerate}
\usepackage{dsfont}
\usepackage{graphicx}
\usepackage{authblk}
\usepackage{hyperref}
\usepackage{epstopdf}
\usepackage{theoremref}
\usepackage{subfig}
\usepackage{bm}
\usepackage{bbm}
\usepackage{diagbox}
\usepackage{tikz}
\usepackage{cite}
\usepackage{blkarray, bigstrut}
\usepackage{booktabs}
\tikzstyle{vertex}=[circle, draw, inner sep=0pt, minimum size=6pt]

\usepackage{algorithm}
\usepackage{algpseudocode}

\newtheorem{theorem}{Theorem}

\newtheorem{lemma}[theorem]{Lemma}

\newtheorem{remark}[theorem]{Remark}
\newtheorem{example}{Example}

\newtheorem{corollary}[theorem]{Corollary}

\newcommand{\C}{\mathbb C}

\newcommand{\Z}{\mathbb Z}
\newcommand{\F}{\mathbb F}

\makeatletter
\newtheorem*{rep@theorem}{\rep@title}
\newcommand{\newreptheorem}[2]{%
\newenvironment{rep#1}[1]{%
 \def\rep@title{#2 \ref{##1}}%
 \begin{rep@theorem}}%
 {\end{rep@theorem}}}
\makeatother
\newreptheorem{theorem}{Theorem}

\newcommand{\ket}[1]{|{#1}\rangle}
\newcommand{\bra}[1]{\langle{#1}|}
\newcommand{\llbr}{[\![}
\newcommand{\rrbr}{]\!]}

\allowdisplaybreaks

\title{Divisible Codes for Quantum Computation}
\author{
	\IEEEauthorblockN{Jingzhen Hu$^\ast$, Qingzhong Liang$^\ast$, and Robert Calderbank} 
	\thanks{$^\ast$The first two authors contributed equally to this work.}
	\thanks{Jingzhen Hu, Qingzhong Liang, and Robert Calderbank are with the Department of Mathematics, Duke University, Durham, NC, USA. E-mail: \{jingzhen.hu, qingzhong.liang, robert.calderbank\}@duke.edu}
}

\newif\ifnotes
\notesfalse

\begin{document}
	\maketitle
	
	\begin{abstract}
	Divisible codes are defined by the property that codeword weights share a common divisor greater than one. They are used to design signals for communications and sensing, and this paper explores how they can be used to protect quantum information as it is transformed by logical gates. Given a CSS code $\mathcal{C}$, we derive conditions that are both necessary and sufficient for a transversal diagonal physical operator $U_Z$ to preserve $\mathcal{C}$ and induce $U_L$. The group of $Z$-stabilizers in a CSS code $\mathcal{C}$ is determined by the dual of a classical $[n, k_1]$ binary code $\mathcal{C}_1$, and the group of $X$-stabilizers is determined by a classical $[n, k_2]$ binary code $\mathcal{C}_2$ that is contained in $\mathcal{C}_1$. The requirement that a diagonal physical operator $U_Z$ fixes a CSS code $\mathcal{C}$ leads to constraints on the congruence of weights in cosets of $\mathcal{C}_2$. These constraints are a perfect fit to divisible codes, and represent an opportunity to take advantage of the extensive literature on classical codes with two or three weights. We construct new families of CSS codes using cosets of the first order Reed Muller code defined by quadratic forms. We provide a simple alternative to the standard method of deriving the coset weight distributions (based on Dickson normal form) that may be of independent interest. Finally, we develop an approach to circumventing the Eastin-Knill Theorem which states that no QECC can implement a universal set of logical gates through transversal gates alone. The essential idea is to design stabilizer codes in layers, with $N_1$ inner qubits and $N_2$ outer qubits, and to assemble a universal set of fault tolerant gates on the inner qubits. 
	\end{abstract}
	
	\section{Introduction}
    \label{sec:intro}
    Quantum error-correcting codes (QECCs) protect information as it is transformed by logical gates. The aim of fault tolerance motivates designing QECCs that implement logical gates through \emph{transversal} physical gates \cite{gottesman1997stabilizer} that are tensor products of unitary operators on individual code blocks.

    We investigate diagonal physical gates that preserve CSS codes \cite{Calderbank-physreva96,calderbank1997quantum,calderbank1998quantum} inducing logical gates. Hence we work with the resolution of the identity determined by the roup of $Z$-stabilizers. 
    The group of $Z$-stabilizers in a CSS code $\mathcal{C}$ is determined by the dual of a classical $[n, k_1]$ binary code $\mathcal{C}_1$, the signs are given by a character of this group, and are determined by a binary vector $\bm{y}$, as described in Section \ref{subsec:perlim_stab}. The group of X-stabilizers is determined by a classical $[n, k_2]$ binary code $\mathcal{C}_2$ that is contained in $\mathcal{C}_1$, and we denote the CSS code $\mathcal{C}$ by CSS$(X, \mathcal{C}_2; Z, C_1^\perp; \bm{y})$. 

    A diagonal physical gate $U_Z$ has $2^n$ diagonal entries $d_{\bm{v}}$, each indexed by a binary vector $\bm{v}$ of length $n$. We introduce our \emph{Generator Coefficient Framework} in Section \ref{sec:GCs}, and use it to show that $U_Z$ preserves the CSS code $\mathcal{C}$ if and only if diagonal entries $d_{\bm{v}}$ and $d_{\bm{w}}$ are identical whenever $\bm{v}$ and $\bm{w}$ belong to the same coset of $\mathcal{C}_2$ in $\mathcal{C}_1+\bm{y}$. By constraining only $2^{k_1}$ out of $2^n$ diagonal entries (those indexed by $\mathcal{C}_1+\bm{y}$) we are able to guarantee not only that $U_Z$ preserves the codespace, but also to completely specify the induced logical operator. The effect of changing the signs of $Z$-stabilizers (changing the vector $\bm{y}$) is to shift the subset of diagonal entries that specify the interaction of $U_Z$ with the CSS code $\mathcal{C}$.

    The approach taken in prior work is to fix a transversal physical gate $U_Z$, and a target logical gate $U_L$, then to derive \emph{sufficient} conditions on a CSS code $\mathcal{C}$ for which $U_Z$ preserves $\mathcal{C}$ and induces $U_L$. The families of triorthogonal codes \cite{bravyi2005universal,bravyi2012magic} and quantum Reed Muller codes \cite{campbell2012magic,landahl2013complex} were derived in this way. In contrast, we fix a CSS code, and use our generator code framework to assemble all possible diagonal physical gates that induce a target logical gate. Our approach supports fault-tolerant architecture by enabling systematic analysis of the locality of these diagonal physical gates. Note that locality is not affected by changing the signs of $Z$-stabilizers. Our approach has the advantage that we are able to derive conditions on a CSS code $\mathcal{C}$ that are both \emph{necessary and sufficient} for a diagonal physical operator $U_Z$ to preserve $\mathcal{C}$ and induce $U_L$. In Section \ref{sec:GCs} we apply the generator coefficient framework to the broad class of Quadratic Form Diagonal (QFD) gates \cite{Rengaswamy-pra19,rengaswamy2020optimality}, which include as a special case the physical gates considered in prior work. We characterize all CSS codes, determined by classical Reed Muller codes, that are fixed by transversal $Z$-rotations through an angle $\pi/2^l$. 

    When the target logical operator is the identity, the physical gates preserving the CSS code represent idling noise to which the code is oblivious. Coherent noise, with identical $Z$-rotation angles on each physical qubit, is a significant source of error in ion trap quantum computers. A CSS code is oblivious to this type of coherent noise if and only if all Hamming weights in $\mathcal{C}_1+\bm{y}$ are identical (see \cite{hu2022mitigating,hu2021css} for more details). We expect the generator coefficient perspective to be useful in designing CSS codes that mitigate the effects of other systems impairments.
    
    The defining property of a classical divisible code \cite{Ward} is that codeword weights share a common divisor greater than one. Divisible codes appear in signal design for wireless communication, in coded radar and sonar, and in the generation of pseudorandom sequences for stream ciphers and for secure authentication (see \cite{golomb2005signal} for more details). In resilient quantum computation, the requirement that a diagonal physical operator $U_Z$ fixes a CSS$(X, \mathcal{C}_2; Z, C_1^\perp; \bm{y})$ code leads to constraints on the congruence of weights in cosets of $\mathcal{C}_2$ \cite{zeng2008semi,haah2018towers}. 

    It is simple to manage congruence of weights in classical codes when the number of weights is small, and our approach creates an opportunity to design CSS codes by taking advantage of the extensive literature on classical codes with two or three weights (see \cite{delsarte1973four,calderbank1986geometry,kohnert2007constructing,ding2015class,kiermaier2020lengths,kurz2021divisible}). In Section IV we construct new families of CSS codes using cosets of the first order Reed Muller code defined by quadratic forms. We also provide a simple alternative to the standard method of deriving coset weight distributions based on Dickson normal form (see \cite[Chapter 15]{macwilliams1977theory}). 

    The Eastin-Knill Theorem \cite{eastin2009restrictions,zeng2011transversality} reveals that no QECC can implement a universal set of logical gates through transversal gates alone. However, several approaches have been developed to circumvent this restriction. 

    Magic state injection circumvents this restriction by consuming magic states to implement non-Clifford gates. State injection \cite{gottesman1999demonstrating,zhou2000methodology} is usually accomplished through magic state distillation (MSD) \cite{bravyi2005universal,reichardt2005quantum,bravyi2012magic,anwar2012qutrit,campbell2012magic,landahl2013complex,campbell2017unified,haah2018codes,krishna2019towards,vuillot2019quantum}, which synthesizes high-fidelity magic states from multiple low-fidelity states. MSD protocols employ CSS codes where a diagonal physical gate induces a fault-tolerant non-Clifford logical gate \cite{gottesman1999demonstrating}. In this context, Bravyi and Haah \cite{bravyi2012magic} introduced the class of triorthogonal codes, where a transversal physical $T$ gate induces a transversal logical $T$ gate up to some logical Clifford gates. Generator coefficients enable analysis of more general pairings of physical and logical gates, for example the hybrid codes introduced by Vasmer and Kubica \cite{vasmer2021morphing}.

    The set of logical operators induced by transversal circuits on two different codes can be universal, and this motivates methods of switching fault-tolerantly between the two code spaces. Hill et al. \cite{hill2013fault} proposed switching between the 5-qubit stabilizer code and the Steane code, while Anderson et al. \cite{anderson2014fault} proposed switching between the Steane code and Reed-Muller codes. An alternative perspective on code switching is to implement the logical operator by fixing gauges on subsystem codes (see Paetznick and Reichardt \cite{paetznick2013universal}, and Bomb\'in \cite{bombin2015gauge}), where gauge qubits are required to do intermediate error corrections). 

    Jochym-O’Connor and R. Laflamme \cite{jochym2014using} proposed to implement a universal set of logical gates on a concatenated code by combining non-transversal physical gates with fault-tolerant recovery operations. Although the gates are not transversal on the concatenated code, the component codes do need to realize a complementary logical gate through transversal gates.

    Finally, Knill et al. \cite{knill1996accuracy} introduced a fault-tolerant controlled-Phase gate on the Steane code by decomposing a non-transversal circuit into pieces, and performing rounds of intermediate error correction to ensure fault-tolerance. Yoder et al. \cite{yoder2016universal} extended this idea of pieceable fault-tolerance to general codes, using the Toffoli and controlled-controlled-$Z$ gates to assemble a universal set of gates. They were able to introduce error correction by decomposing the non-transversal circuit, identifying the stabilizer group corresponding to the intermediate states, and measuring the stabilizers.

    Section V describes the design of stabilizer codes in layers, with $N_1$ inner qubits and $N_2$ outer qubits, with the aim of assembling a universal set of fault tolerant gates on the inner qubits. The $N_1$ inner qubits are the logical qubits of an outer CSS code on $N_2$ qubits. A transversal diagonal physical gate on $N_2$ qubits preserves the outer code, inducing a target logical gate on the inner code. The general design is introduced through an example taken from Section \ref{sec:GCs}, where the outer code is a $\llbr 31, 5, 3\rrbr$ CSS code fixed by a transversal $T$ gate. The inner code is the $\llbr 5, 1, 3\rrbr$ stabilizer code and the induced logical operator is a $T$ gate on the inner qubit. It may be useful to view this combination of codes as a factorization of a $\llbr 31, 1, 3 \rrbr$ triorthogonal code. Factorization is similar to code switching, concatenation, and pieceable fault tolerance in that it can produce a universal set of gates without requiring teleportation of magic states. The $\llbr 31,5,3\rrbr$ outer code belongs to a new $\llbr n=2^m-1,1\le k\le 1+\sum_{i=1}^{m-4}(m-i),d=3 \rrbr$ CSS code family, designed from the congruence of weights in cosets of quadratic forms. When $m=6$, the $\llbr 63,7,3\rrbr$ CSS code can serve as the outer code of the $\llbr 7,1,3\rrbr$ Steane code to support a fault-tolerant $T$ gate. However, inducing a logical gate on the inner code does require encoding/decoding algorithms to pass between inner and outer codes. The overhead of factorization depends on the complexity of these algorithms.

	\section{Preliminaries and Notation}\label{sec_prelim}
    \subsection{The Pauli Group}
        \label{subsec:perlim_Pauli}
    Let $\imath\coloneqq \sqrt{-1}$ be the imaginary unit. 
    Any $2\times 2$ Hermitian matrix can be uniquely expressed as a real linear combination of the four single qubit Pauli matrices/operators
    \begin{align}
    I_2 \coloneqq \begin{bmatrix}
    1 & 0\\
    0 & 1
    \end{bmatrix},~  
    X \coloneqq \begin{bmatrix} 
    0 & 1\\
    1 & 0
    \end{bmatrix},~  
    Z \coloneqq  \begin{bmatrix} 
    1 & 0\\
    0 & -1
    \end{bmatrix}, 
    \end{align}
    and $Y\coloneqq \imath XZ$.
    The operators satisfy 
    $
    X^2= Y^2= Z^2=I_2,~  X Y=- Y X,~  X Z=- Z X,~ \text{ and }  Y Z=- Z Y.
    $
    
    Let $\F_2 = \{0,1\}$ denote the binary field. Let $n\ge 1$ and $N=2^n$. Let $A \otimes B$ denote the Kronecker product (tensor product) of two matrices $A$ and $B$. Given binary vectors $\bm{a}=[a_1,a_2,\dots,a_n]$ and $\bm{b}=[b_1,b_2,\dots,b_n]$ with $a_i,b_j =0$ or $1$, we define the operators
    \begin{align}
    D(\bm{a},\bm{b})&\coloneqq X^{a_1} Z^{b_1}\otimes \cdots \otimes  X^{a_n} Z^{b_n},\\
    E(\bm{a},\bm{b}) &\coloneqq\imath^{\bm{a}\bm{b}^T \bmod 4}D(\bm{a},\bm{b}).
    \end{align}
    
    Note that $D(\bm{a},\bm{b})$ can have order $1,2$ or $4$, but $E(\bm{a},\bm{b})^2=\imath^{2\bm{a}\bm{b}^T}D(\bm{a},\bm{b})^2=\imath^{2ab^T}( \imath^{2\bm{a}\bm{b}^T} I_N)=I_N$. 
    We define the $n$-qubit \textit{Pauli group}
    \begin{equation}
    \mathcal{P}_N \coloneqq\{\imath^\kappa D(\bm{a},\bm{b}): \bm{a},\bm{b}\in \F_2^n, \kappa\in \Z_{4} \},
    \end{equation}
    where $\Z_{2^l} = \{0,1,\dots,2^l-1\}$.
    The $n$-qubit Pauli matrices form an orthonormal basis for the vector space of $N\times N$ complex matrices ($\C^{N\times N}$) under the normalized Hilbert-Schmidt inner product $\langle A,B\rangle \coloneqq \mathrm{Tr}(A^\dagger B)/N$ \cite{gottesman1997stabilizer}. 
    
    
    The symplectic inner product is $\langle [\bm{a},\bm{b}],[\bm{c},\bm{d}]\rangle_S=\bm{a}\bm{d}^T+\bm{b}\bm{c}^T \bmod 2$. Since $ X Z=- Z X$, we have 
    \begin{equation}
    E(\bm{a},\bm{b})E(\bm{c},\bm{d})=(-1)^{\langle [\bm{a},\bm{b}],[\bm{c},\bm{d}]\rangle_S}E(\bm{c},\bm{d})E(\bm{a},\bm{b}).
    \end{equation}
    
    We use the \textit{Dirac notation}, $|\cdot \rangle$ to represent the basis states of a single qubit in $\C^2$. For any $\bm{v}=[v_1,v_2,\cdots, v_n]\in \F_2^n$, we define $|\bm{v}\rangle=|v_1\rangle\otimes|v_2\rangle\otimes\cdots\otimes|v_n\rangle$, the standard basis vector in $\C^N$ with $1$ in the position indexed by $\bm{v}$ and $0$ elsewhere. We write the Hermitian transpose of $|\bm{v}\rangle$ as $\langle \bm{v}|=|\bm{v}\rangle^\dagger$. 

    \subsection{The Clifford Hierarchy} 
    \label{subsec:perlim_Clifford}
    The \textit{Clifford hierarchy} of unitary operators was introduced in \cite{gottesman1999demonstrating}. The first level of the hierarchy is defined to be the Pauli group $\mathcal{C}^{(1)}=\mathcal{P}_N$. For $l\ge 2$, the levels $l$ are defined recursively as 
    \begin{equation}\label{eqn:def_Cliff_hierarchy}
    \mathcal{C}^{(l)}:=\{U\in \mathbb{U}_N: U \mathcal{P}_N U^\dagger\subset \mathcal{C}^{(l-1)}\},
    \end{equation}
    where $\mathbb{U}_N$ is the group of $N\times N$ unitary matrices. The second level is the Clifford Group~\cite{Gottesman-icgtmp98}, $\mathcal{C}^{(2)}$, which can be generated using the unitaries \textit{Hadamard}, \textit{Phase}, and either of \textit{Controlled-NOT} (C$X$) or \textit{Controlled-$Z$} (C$Z$) defined respectively as
    \begin{align}
    H\coloneqq\frac{1}{\sqrt{2}}
    \begin{bmatrix}
    1 & 1\\
    1 & -1
    \end{bmatrix},~
    P\coloneqq\begin{bmatrix}
    1 & 0\\
    0 & \imath
    \end{bmatrix},~
    \text{C}Z_{ab}  \coloneqq \ket{0}\bra{0}_a \otimes (I_2)_b + \ket{1}\bra{1}_a \otimes Z_b,~
    \text{C}X_{a \rightarrow b}  \coloneqq \ket{0}\bra{0}_a \otimes (I_2)_b + \ket{1}\bra{1}_a \otimes X_b.
    \end{align}
    The Clifford group in combination with \emph{any} unitary from a higher level can be used to approximate any unitary operator arbitrarily well~\cite{boykin1999universal,nebe2001invariants}. 
    Hence, they form a universal set for quantum computation. A widely used choice for the non-Clifford unitary is the $T$ gate in the third level defined by $T=Z^{\frac{1}{4}}\equiv e^{-\frac{\imath\pi}{8} Z}.$
    
    Let $\mathcal{D}_N$ be the $N\times N$ diagonal matrices, and $\mathcal{C}^{(l)}_d \coloneqq \mathcal{C}^{(l)} \cap \mathcal{D}_N$. The diagonal gates at each level in the hierarchy form a group, but for $l \ge 3$, the gates in $\mathcal{C}^{(l)}$ no longer form a group.
    Note that $\mathcal{C}^{(l)}_d$ can be generated using the ``elementary" unitaries  C$^{(0)}Z^{\frac{1}{2^{l}}}$, C$^{(1)}Z^{\frac{1}{2^{l-1}}}, \dots, $C$^{(l-2)}Z^{\frac{1}{2}}$, C$^{(l-1)} Z$ \cite{zeng2008semi}, where C$^{(i)}Z^{\frac{1}{2^j}} \coloneqq \sum_{\bm{u}\in \F_2^{i+1}}\ket{\bm{u}}\bra{\bm{u}} +e^{\imath \frac{\pi}{2^j}}\ket{\bm{1}}\bra{\bm{1}}$ and $\bm{1}\in\F_2^{i+1}$ denotes the vector with every entry $1$. In general, the length of $\bm{1}$ can change and should be clear in the context.

\subsection{Stabilizer Codes}
    \label{subsec:perlim_stab}
    We define a stabilizer group $\mathcal{S}$ to be a commutative subgroup of the Pauli group $\mathcal{P}_N$, where every group element is Hermitian and no group element is $-I_N$. We say $\mathcal{S}$ has dimension $r$ if it can be generated by $r$ independent elements as $\mathcal{S}=\langle \nu_i E(\bm{c_i},\bm{d_i}): i=1,\dots, r \rangle$, where $\nu_i\in\{\pm1\}$ and $\bm{c_i},\bm{d_i}\in \F_2^n$. Since $\mathcal{S}$ is commutative, we must have $\langle [\bm{c_i},\bm{d_i}],[\bm{c_j},\bm{d_j}]\rangle_S=\bm{c_i}\bm{d_j}^T+\bm{d_i}\bm{c_j}^T=0\bmod 2$.
    
    Given a stabilizer group $\mathcal{S}$, the corresponding \textit{stabilizer code} \cite{Calderbank-physreva96,calderbank1997quantum,gottesman1997stabilizer,calderbank1998quantum} is the fixed subspace $\mathcal{V}(\mathcal{S)}:=\{|\psi\rangle \in \C^N: g|\psi\rangle=|\psi\rangle \text{ for all } g\in \mathcal{S} \}$. 
    We refer to the subspace $\mathcal{V}(\mathcal{S})$ as an $\left[\left[n,k,d\right]\right]$ stabilizer code because it encodes $k:=n-r$ \text{logical} qubits into $n$ \textit{physical} qubits. The minimum distance $d$ is defined to be the minimum weight of any operator in $\mathcal{N}_{\mathcal{P}_N}\left(\mathcal{S}\right)\setminus \mathcal{S}$. Here, the weight of a Pauli operator is the number of qubits on which it acts non-trivially (i.e., as $ X,~Y$ or $ Z$), and $\mathcal{N}_{\mathcal{P}_N}\left(\mathcal{S}\right)$ denotes the normalizer of $\mathcal{S}$ in $\mathcal{P}_N$. 
    
    For any Hermitian Pauli matrix $E\left(\bm{c},\bm{d}\right)$ and $\nu\in\{\pm 1\}$, the operator $\frac{I_N+\nu E\left(\bm{c},\bm{d}\right)}{2}$ projects onto the $\nu$-eigenspace of $E\left(\bm{c},\bm{d}\right)$. Thus, the projector onto the codespace $\mathcal{V}(\mathcal{S})$ of the stabilizer code defined by $\mathcal{S}=\langle \nu_i E\left(\bm{c_i},\bm{d_i}\right): i=1,\dots, r \rangle$ is 
    \begin{equation}
    \Pi_{\mathcal{S}}=\prod_{i=1}^{r}\frac{\left(I_N+\nu_i E\left(\bm{c_i},\bm{d_i}\right)\right)}{2}=\frac{1}{2^r}\sum_{j=1}^{2^r}\epsilon_j E\left(\bm{a_j},\bm{b_j}\right),
    \end{equation}
    where $\epsilon_j\in \{\pm 1 \}$ is a character of the group $\mathcal{S}$, and is determined by the signs of the generators that produce $E(\bm{a_j},\bm{b_j})$: $\epsilon_jE\left(\bm{a_j},\bm{b_j}\right)=\prod_{t\in J\subset \{1,2,\dots,r\} } \nu_t E\left(\bm{c_t},\bm{d_t}\right)$ for a unique $J$.
    
    A \textit{CSS code} is a particular type of stabilizer code with generators that can be separated into strictly $X$-type and strictly $Z$-type operators. Consider two classical binary codes $\mathcal{C}_1,\mathcal{C}_2$ such that $\mathcal{C}_2\subset \mathcal{C}_1$, and let $\mathcal{C}_1^\perp$, $\mathcal{C}_2^\perp$ denote the dual codes. Note that $\mathcal{C}_1^\perp\subset \mathcal{C}_2^\perp$. Suppose that $\mathcal{C}_2 = \langle \bm{c_1},\bm{c_2},\dots,\bm{c_{k_2}} \rangle$ is an $[n,k_2]$ code and $\mathcal{C}_1^\perp =\langle \bm{d_1},\bm{d_2}\dots,\bm{d_{n-k_1}}\rangle$ is an $[n,n-k_1]$ code.
    Then, the corresponding CSS code has the stabilizer group 
	\begin{align*}
	\mathcal{S} 
	&=\langle \nu_{(\bm{c_i},\bm{0})} E\left(\bm{c_i},\bm{0}\right), \nu_{(\bm{0},\bm{d_j})} E\left(\bm{0},\bm{d_j}\right)\rangle_{i=1;~j=1}^{i=k_2;~j=n-k_1} \nonumber\\
	&=\{\epsilon_{(\bm{a},\bm{0})} \epsilon_{(\bm{0},
	\bm{b})} E\left(\bm{a},\bm{0}\right)E\left(\bm{0},\bm{b}\right): \bm{a}\in \mathcal{C}_2, \bm{b}\in \mathcal{C}_1^\perp\}, 
	\end{align*}
	where $\nu_{(\bm{c_i},\bm{0})},\nu_{(\bm{0},\bm{d_j})},\epsilon_{(\bm{a},\bm{0})},\epsilon_{(\bm{0},
	\bm{b})}  \in\{\pm 1 \}$.
	We capture sign information through character vectors 
	$\bm{y}\in \F_2^n/\mathcal{C}_1,\bm{r} \in \F_2^n/\mathcal{C}_2^\perp$ such that for any $ \epsilon_{(\bm{a},\bm{0})} \epsilon_{(\bm{0},\bm{b})} E\left(\bm{a},\bm{0}\right)E\left(\bm{0},\bm{b}\right) \in S$, we have $ \epsilon_{(\bm{a},\bm{0})} = (-1)^{\bm{a}\bm{r}^T}$ and $ \epsilon_{(\bm{0},\bm{b})} = (-1)^{\bm{b}\bm{y}^T}$. 
	If $\mathcal{C}_1$ and $\mathcal{C}_2^\perp$ can correct up to $t$ errors, then $S$ defines an $\left[\left[n, k=k_1-k_2, d\right]\right]$ CSS code with $d\ge 2t+1$, which we will represent as CSS($X,\mathcal{C}_2,\bm{r};Z,\mathcal{C}_1^\perp, \bm{y}$). If $G_2$ and $G_1^\perp$ are the generator matrices for $\mathcal{C}_2$ and $\mathcal{C}_1^\perp$ respectively, then the $(n-k_1+k_2)\times (2n)$ matrix
	\begin{equation}
	G_{\mathcal{S}}=\left[\begin{array}{c|c}
	G_2 &  \\ \hline
	&  G_1^\perp
	\end{array} \right]
	\end{equation}
	generates $\mathcal{S}$. 
	The encoding map $g_e:\ket{\bm{\alpha}}_L\in \mathbb{F}_2^k \to \ket{\overline{\bm{\alpha}}}\in \mathcal{V}(\mathcal{S})$ of a CSS($X,\mathcal{C}_2,\bm{r};Z,\mathcal{C}_1^\perp, \bm{y}$) code can be written as \cite{hu2021designing} 
		\begin{equation} \label{eqn:gen_encode_map}
		\ket{\overline{\bm{\alpha}}}
		\coloneqq \frac{1}{\sqrt{|\mathcal{C}_2|}} \sum_{\bm{x}\in \mathcal{C}_2} (-1)^{\bm{x}\bm{r}^T}\ket{\bm{\alpha}  G_{\mathcal{C}_1/\mathcal{C}_2} \oplus \bm{x} \oplus \bm{y}}.
	\end{equation}
	We define the generating set $\{X^L_j, Z^L_j \in \mathcal{HW}_{2^k} : j = 1, \dots k=k_1-k_2\}$ for the logical Pauli operators by the actions 
	\begin{align}
	X_j^L \ket{\bm{\alpha}}_L = \ket{\bm{\alpha'}}_L, \text{ where } \alpha'_i = 
	\left\{ \begin{array}{lc}
	\alpha_i,  &\text{ if } i \neq j, \\
	\alpha_i \oplus 1, & \text{ if } i = j,
	\end{array}\right.
	\end{align}
	and $Z_j^L \ket{\bm{\alpha}}_L = (-1)^{\alpha_j}\ket{\bm{\alpha}}_L.$
	Let $\bar{X}_j, \bar{Z}_j$ be the $n$-qubit operators which are physical representatives of $X_j^L, Z^L_j$ for $j= 1,\dots,k $. Then $\bar{X}_j, \bar{Z}_j$ commute with the stabilizer group $S$ and satisfy
	\begin{align}\label{eqn:logical_condition}
	\bar{X}_i \bar{Z}_j = 	
	\left\{ 
	\begin{array}{lc}
	\bar{Z}_j \bar{X}_i,  &\text{ if } i \neq j, \\
	-\bar{Z}_j \bar{X}_i, &\text{ if } i = j.
	\end{array}\right. 
	\end{align}
	Given the choice of $G_{\mathcal{C}_1/\mathcal{C}_2}$, there exists a unique set of vectors $\{\bm{\gamma_1},\cdots,\bm{\gamma_k} \in \mathcal{C}_2^\perp: G_{\mathcal{C}_1/\mathcal{C}_2}\bm{\gamma_i} = \bm{e_i} \text{ for all } i = 1,\dots,k \}$, where $\{\bm{e_i}\}_{i =1,\dots,k}$ is the standard basis of $\F_2^k$.
	If $\bm{\gamma_i}$ is the $i$-the row of generator matrix $G_{\mathcal{C}_2^\perp / \mathcal{C}_1^\perp}$, then 
	\begin{equation}\label{eqn:log_req}
	G_{\mathcal{C}_1/\mathcal{C}_2} G_{\mathcal{C}_2^\perp /\mathcal{C}_1^\perp }^T = I_{k}.
	\end{equation}
	Assume $\bm{w_i}, \bm{\gamma_i}$ are the $i$-th rows of the above coset generator matrices $G_{\mathcal{C}_1/\mathcal{C}_2}$, $G_{\mathcal{C}_2^\perp / \mathcal{C}_1^\perp}$ respectively. 
	Then, we can choose 
		$\bar{X}_{i} = E(\bm{w_i},\bm{0})$
	and 
	\begin{equation}\label{eqn:css_log_paulis}
		\bar{Z}_{i} = (-1)^{\bm{\gamma_i}\bm{y}^T}E(\bm{0},\bm{\gamma_i}) 
		\coloneqq \epsilon_{(\bm{0},\bm{\gamma_i})}E(\bm{0},\bm{\gamma_i}).
	\end{equation}

	Since we consider diagonal gates, the signs of $X$-stabilizers do not matter, and we assume $\bm{r}=\bm{0}$ throughout this paper. 
    
    \section{Generator Coefficients}\label{sec:GCs}
We review the \emph{Generator Coefficient Framework} which describes the evolution of stabilizer code states under a physical diagonal gate $U_Z = \sum_{\bm{u}\in \F_2^n} d_{\bm{u}}\ket{\bm{u}}\bra{\bm{u}}$ (See \cite{hu2021designing} for more details) under matrix representation. Note that $\ket{\bm{u}}\bra{\bm{u}} = \frac{1}{2^n}\sum_{\bm{v}\in\F_2^n} (-1)^{\bm{uv}^T} E(\bm{0},\bm{v})$. Alternatively, we may expand $U_Z$ in the Pauli basis 
    \begin{align}
        U_Z = \sum_{\bm{v}\in\F_2^n} f(\bm{v}) E(\bm{0},\bm{v}),
    \end{align}
    where 
    \begin{align}\label{eqn:coeff_UZ}
        f(\bm{v}) = \frac{1}{2^n} \sum_{\bm{u}\in\F_2^n} (-1)^{\bm{uv}^T} d_{\bm{u}}. 
    \end{align}
    The Hadamard gate $H_{2^n}$ connects the coefficients in the standard basis with those in the Pauli basis as follows
    \begin{align}
        \left[f(\bm{v})\right]_{\bm{v}\in\F_2^n} =  \left[d_{\bm{u}}\right]_{\bm{u}\in\F_2^n} H_{2^n}, 
    \end{align}
    where $H=\frac{1}{\sqrt{2}}\left(\ket{0}\bra{0}+\ket{0}\bra{1}+\ket{1}\bra{0}-\ket{1}\bra{1}\right)$ and $H_{2^n} = H \otimes H_{2^{n-1}} = H^{\otimes n}$ is the Hadamard gate.
    
    We consider the average logical channel induced by $U_Z$ on an $\llbr n,k,d\rrbr$ CSS($X,\mathcal{C}_2;Z,\mathcal{C}_1^\perp, \bm{y}$) code resulting from the four steps
    : (1) preparing any code state $\rho_1$; (2) applying a diagonal physical gate $U_Z$ to obtain $\rho_2$; (3) using $X$-stabilizers to measure $\rho_2$ (we only consider $Z$-errors as the same reasons in \cite{bravyi2005universal,bravyi2012magic}), to obtain the syndrome $\bm{\mu}$ with probability $p_{\bm{\mu}}$, and the post-measurement state $\rho_3$; (4) applying a Pauli correction to $\rho_3$, to obtain $\rho_4$. The correction might induce some undetectable $Z$-logical $\epsilon_{(\bm{0},\bm{\gamma_{\mu}})}E(\bm{0},\bm{\gamma_{\mu}})$ with $\bm{\gamma}_{\bm{0}}=\bm{0}$. 
    Let $B_{\bm{\mu}}$ be the effective physical operator corresponding to the syndrome $\bm{\mu}$. Then the evolution of code states can be described as 
    \begin{align}
        \rho_4= \sum_{\bm{\mu}\in \F_2^n/\mathcal{C}_2^\perp} B_{\bm{\mu}} \rho_1 B_{\bm{\mu}}^\dagger. 
    \end{align}
    The generator coefficients $A_{\bm{\mu},\bm{\gamma}}$ are obtained by expanding the logical operator $B_{\bm{\mu}}$ in terms of $Z$-logical Pauli operators $\epsilon_{(\bm{0},\bm{\gamma})}E(\bm{0},\bm{\gamma})$,
    \begin{align}\label{eqn:kraus_ops}
    B_{\bm{\mu}}=\epsilon_{(\bm{0},\bm{\gamma_{\mu}})}E(\bm{0},\bm{\gamma_{\mu}}) \sum_{\bm{\gamma}\in \mathcal{C}_2^\perp / \mathcal{C}_1^\perp} A_{\bm{\mu},\bm{\gamma}}~ \epsilon_{(\bm{0},\bm{\gamma})}E(\bm{0},\bm{\gamma}),
    \end{align}
    where $\epsilon_{(\bm{0},\bm{\gamma_{\mu}})}E(\bm{0},\bm{\gamma_{\mu}})$ models the $Z$-logical Pauli correction introduced by a decoder. 
    For each pair of an $X$-syndrome $\bm{\mu} \in \F_2^n / \mathcal{C}_2^\perp$ and a $Z$-logical $\bm{\gamma} \in \mathcal{C}_2^\perp / \mathcal{C}_1^\perp$, the generator coefficient $A_{\bm{\mu},\bm{\gamma}}$ corresponding to $U_Z$ is
    \begin{equation}\label{eqn:def_GC_UZ}
	    A_{\bm{\mu},\bm{\gamma}}\coloneqq \sum_{\bm{z}\in \mathcal{C}_1^\perp+\bm{\mu}+\bm{\gamma}}\epsilon_{(\bm{0},\bm{z})}f(\bm{z}),
	\end{equation}
	    where $\epsilon_{(\bm{0},\bm{z})} = (-1)^{\bm{z}\bm{y}^T}$ is the sign of the $Z$-stabilizer $E(\bm{0},\bm{z})$. 
	   The chosen $Z$-logicals and $X$-syndromes are not unique, but different choices only differ by a global phase. Generator coefficients use the CSS code to organize the Pauli coefficients of $U_Z$ into groups and to balance them by tuning the signs of $Z$-stabilizers. 
	We use \eqref{eqn:coeff_UZ} to simplify \eqref{eqn:def_GC_UZ} as
	\begin{align}
	    A_{\bm{\mu},\bm{\gamma}}
	    &= \frac{1}{2^n} \sum_{\bm{u}\in\F_2^n} \sum_{\bm{z}\in \mathcal{C}_1^\perp+\bm{\mu}+\bm{\gamma}} (-1)^{\bm{zy}^T} (-1)^{\bm{zu}^T} d_{\bm{u}} \nonumber\\
	    &=\frac{1}{|\mathcal{C}_1|}\sum_{\bm{u}\in \mathcal{C}_1} (-1)^{(\bm{\mu}\oplus\bm{\gamma})\bm{u}^T}  d_{\bm{u}\oplus\bm{y}},\label{eqn:GC_C1_y}
	\end{align}
	where $|\mathcal{C}_1| = 2^{k_1}$ is the size of $\mathcal{C}_1$. 
	We organize the generator coefficients in a matrix $M_{(\F_2^n/\mathcal{C}_2^\perp,\mathcal{C}_2^\perp/\mathcal{C}_1^\perp)}$ with rows indexed by $X$-syndromes and columns by $Z$-logicals,
	\begin{align}\label{eqn:GC_matrix_rep_mu&gamma}
	M_{(\F_2^n/\mathcal{C}_2^\perp,\mathcal{C}_2^\perp/\mathcal{C}_1^\perp)} = 
	\begin{bmatrix}
	[A_{\bm{\mu}=\bm{0},\bm{\gamma}}]_{\bm{\gamma}\in\mathcal{C}_2^\perp/\mathcal{C}_1^\perp}\\
	[A_{\bm{\mu}=\bm{\mu_1},\bm{\gamma}}]_{\bm{\gamma}\in\mathcal{C}_2^\perp/\mathcal{C}_1^\perp} \\
	\vdots\\
	[A_{\bm{\mu}=\bm{\mu_{2^{k_2}-1}},\bm{\gamma}}]_{\bm{\gamma}\in\mathcal{C}_2^\perp/\mathcal{C}_1^\perp}
	\end{bmatrix}_{\bm{\mu}\in \F_2^n/\mathcal{C}_2^\perp}.
	\end{align}
	For fixed $\bm{\mu}\in \F_2^n/\mathcal{C}_2^\perp$,
	\begin{align} \label{eqn:fix_syndrome_GC}
	  [A_{\bm{\mu},\bm{\gamma}}]_{\bm{\gamma}\in\mathcal{C}_2^\perp/\mathcal{C}_1^\perp}
	    = 
	   \frac{1}{|C_1|} [d_{\bm{u}\oplus\bm{y}}]_{\bm{u}\in\mathcal{C}_1} H^{\bm{\mu}}_{(\mathcal{C}_1,\mathcal{C}_2^\perp/\mathcal{C}_1^\perp)},
	 \end{align}
	 where $H^{\bm{\mu}}_{(\mathcal{C}_1,\mathcal{C}_2^\perp/\mathcal{C}_1^\perp)} = [(-1)^{(\bm{\mu}\oplus\bm{\gamma})\bm{u}^T}]_{{\bm{u}\in\mathcal{C}_1},\bm{\gamma}\in \mathcal{C}_2^\perp/\mathcal{C}_1^\perp}$.
	\begin{theorem}[Theorem 7 in \cite{hu2021designing}]\label{thm:general_invar}
	The physical gate $U_Z= \sum_{\bm{u}\in \F_2^n} d_{\bm{u}}\ket{\bm{u}}\bra{\bm{u}}$ preserves a CSS($X,\mathcal{C}_2;Z,\mathcal{C}_1^\perp, \bm{y}$) codespace if and only if 
	\begin{equation}\label{eqn:preserved_by_Uz}
		\sum_{\bm{\gamma}\in \mathcal{C}_2^\perp/\mathcal{C}_1^\perp} |A_{\bm{0},\bm{\gamma}}|^2=
		\sum_{\bm{\gamma}\in \mathcal{C}_2^\perp/\mathcal{C}_1^\perp}\overline{A_{\bm{0},\bm{\gamma}}}A_{\bm{0},\bm{\gamma}}=1.
	\end{equation}
	Here, $|\cdot|$ denotes the complex norm. 
	\end{theorem}
	\begin{proof}
	Invariance of the codespace is equivalent to requiring the effective physical operator corresponding to the trivial syndrome $B_{\bm{\mu}=\bm{0}}$ to be unitary. 
	\end{proof}
	Note that \eqref{eqn:preserved_by_Uz} is also equivalent to $[A_{\bm{\mu}\neq \bm{0},\bm{\gamma}}]_{\bm{\gamma}\in \mathcal{C}_2^\perp/\mathcal{C}_1^\perp} = \bm{0}$ \cite[Theorem 6]{hu2021designing}. 
	The induced logical operator is 
	\begin{align}
	    U_Z^L 
		&= \sum_{\bm{\alpha}\in \F_2^k} A_{\bm{0}, g(\bm{\alpha})} E(\bm{0},\bm{\alpha}) \nonumber \\
		&= \frac{1}{|\mathcal{C}_1|}\sum_{\bm{\alpha}\in \F_2^k} \sum_{\bm{u}\in \mathcal{C}_1} (-1)^{g(\bm{\alpha}) \bm{u}^T}  d_{\bm{u}\oplus\bm{y}} E(\bm{0},\bm{\alpha}),\label{eqn:dig_log_op}
	\end{align}
	where $g: \F_2^k \to \mathcal{C}_2^\perp / \mathcal{C}_1^\perp $ is a bijective map defined by $g(\bm{\alpha}) =\bm{\alpha} G_{\mathcal{C}_2^\perp /\mathcal{C}_1^\perp}$. Here, $G_{\mathcal{C}_2^\perp /\mathcal{C}_1^\perp}$ is one choice of the generator matrix of $Z$-logicals (coset representatives of $\mathcal{C}_2^\perp /\mathcal{C}_1^\perp$). 
	\begin{example}\label{exam:15_1_3}
	The $\llbr 15,1,3\rrbr$ punctured quantum Reed-Muller code \cite{bravyi2005universal} is a CSS($X, \mathcal{C}_2; Z, \mathcal{C}_1^\perp, \bm{y}=\bm{0})$ code, where $\mathcal{C}_2$ is generated by the degree one monomials, $x_1, x_2, x_3, x_4$, and $\mathcal{C}_1^\perp$ = $\langle x_1, x_2, x_3, x_4, x_1x_2, x_1x_3, x_1x_4, x_2x_3, x_2x_4, x_3x_4\rangle$, with the first coordinate removed in both $\mathcal{C}_2$ and $\mathcal{C}_1^\perp$. It's also a triorthogonal code for which a physical transversal $T$ gate, $U_Z = \sum_{\bm{u}\in\F_2^n} \left(e^{\imath\pi/4}\right)^{w_H(\bm{u})}\ket{\bm{u}}\bra{\bm{u}}$, induces a logical transversal $T$ gate up to some Clifford gates. Here, $w_H(\bm{u}) = \bm{u}\bm{u}^T$ denotes the Hamming weight of the binary vector $\bm{u}$. Note that $\mathcal{C}_1$ is the classical punctured RM$(1,4)$ code with weight distribution given in Table \ref{tab:15_1_3_C_1} below.
    \begin{table}[h!]
    \vspace{-10pt}
     \caption{The weight distribution of $\mathcal{C}_1$ for the $\llbr15,1,3\rrbr$ code}
     \renewcommand{\arraystretch}{1.1} 
        \centering
        \begin{tabular}{|c|c|c|c|c|}
        \hline
            weight & 0 & 7 & 8 & 15  \\
            \hline
            multiplicity & 1 & 15 & 15 & 1 \\
            \hline
        \end{tabular}
        \label{tab:15_1_3_C_1}
    \end{table}
    
    Then, $d_{\bm{u}} = 1$ for $\bm{u}\in\mathcal{C}_1$ satisfying $w_H(\bm{u})= 0$ or $8$, and $d_{\bm{u}} = e^{-\imath\pi/4}$ for $\bm{u}\in\mathcal{C}_1$ satisfying $w_H(\bm{u})= 7$ or $15$. Since the $Z$-logical $\bm{\gamma}=\bm{1}$, the all-one vector, it only changes the signs of $d_{\bm{u}}$ with odd weight. It follows from \eqref{eqn:dig_log_op} that the induced logical operator is
    \begin{equation}
        U_Z^L = \frac{16}{32}(1 + e^{-\imath\pi/4})E(0,0) + \frac{16}{32}(1 - e^{-\imath\pi/4})E(0,1) = T^\dagger.
    \end{equation}
    
	\end{example}
	\begin{remark}
	  It follows from \eqref{eqn:dig_log_op} that the induced logical operator is completely specified by $|\mathcal{C}_1|$ diagonal entries in the physical gate $U_Z$. If we choose a CSS code and target a particular logical gate, then the constraints on the corresponding physical gates only apply to the diagonal elements corresponding to the coset $\mathcal{C}_1+\bm{y}$. 
	\end{remark}
	
    Given a CSS code, the generator coefficient framework not only represents when a physical diagonal gate preserves the codespace, but it also characterizes all the possible physical gates that realize a target diagonal logical gate. We start from the simplest case, when the logical operator is the identity. 
    
    \begin{lemma}\label{lemma:log_identity}
    The physical gate $U_Z= \sum_{\bm{u}\in \F_2^n} d_{\bm{u}}\ket{\bm{u}}\bra{\bm{u}}$ acts as the logical identity on the CSS($X,\mathcal{C}_2;Z,\mathcal{C}_1^\perp, \bm{y}$) codespace if and only if $d_{\bm{u}\oplus\bm{y}}$ are the same for all $\bm{u}\in\mathcal{C}_1$.
    \end{lemma}
    \begin{proof}
   It follows from \eqref{eqn:dig_log_op} that $U_Z^L = I_{2^k}$ if and only if 
    \begin{equation}
    \left|A_{\bm{\mu}=\bm{0},\bm{\gamma}=\bm{0}}\right| = \left|\frac{1}{|\mathcal{C}_1|}\sum_{\bm{u}\in \mathcal{C}_1 }  d_{\bm{u}\oplus\bm{y}}\right| = 1,
    \end{equation}
     which is equivalent to requiring that $2^{k_1}$ diagonal entries of the physical gate $U_Z$ indexed by the set $\mathcal{C}_1+\bm{y}$ are identical. 
    \end{proof}
    The mapping from a physical gate that preserves a given CSS code to the induced logical operator is a group homomorphism. The kernel of this homomorphism is the group of phsyical gates that induce the logical identity. 
    \begin{remark}
      Given a CSS code, Lemma \ref{lemma:log_identity} characterizes all the diagonal physical gates that induce the identity on the codespace. This enables code design within a decoherence-free subspace (DFS) for a particular noise system. For homogeneous coherent noise (same angle on each physical qubit), we consider
    \begin{align}
    U_Z 
    =\begin{bmatrix}
    1 & 0\\
    0 & e^{\imath\theta} 
    \end{bmatrix}^{\otimes n} \equiv  \sum_{\bm{u}\in\F_2^n} \left(e^{\imath\theta} \right)^{w_H(\bm{u})}\ket{\bm{u}}\bra{\bm{u}}, 
    \end{align}
    with $\theta\in(0, 2\pi)$. We design CSS codes that are oblivious to all such gates by making sure all the Hamming weights in the coset $\mathcal{C}_1 +\bm{y}$ are the same (a new perspective on the results in \cite{hu2022mitigating,ouyang2021avoiding}). For coherent noise with inhomogeneous angles, this perspective enables code design to mitigate these correlated errors. 
    For example, we consider 
    $U_Z = 
    \begin{bmatrix}
    1 & 0\\
    0 & e^{\imath\theta_1} 
    \end{bmatrix}{\otimes}
    \begin{bmatrix}
    1 & 0\\
    0 & e^{\imath\theta_2} 
    \end{bmatrix}{\otimes}
    \begin{bmatrix}
    1 & 0\\
    0 & e^{\imath\theta} 
    \end{bmatrix}{\otimes}
    \begin{bmatrix}
    1 & 0\\
    0 & e^{\imath\theta} 
    \end{bmatrix}{\otimes}
    \begin{bmatrix}
    1 & 0\\
    0 & e^{\imath\theta^{'}_1} 
    \end{bmatrix}{\otimes}
    \begin{bmatrix}
    1 & 0\\
    0 & e^{\imath\theta_2'} 
    \end{bmatrix}{\otimes}
    $, with $\theta\in(0, 2\pi)$ and $\theta_1 +\theta_2 = \theta_1^{'} +\theta_2^{'} = \theta $.
    By selecting select the diagonal elements of $U_Z$ with the same value, we design a $\llbr 6,1,2\rrbr$ CSS code within a DFS for the inhomogeneous noise system, where  
    \begin{align}
    G_{\mathcal{C}_2} =     
    \begin{bmatrix}
    1 & 1 & 1 & 1 & 1 & 1
    \end{bmatrix}
    \subset
    G_{\mathcal{C}_1} =     
    \begin{bmatrix}
    1 & 1 & 1 & 1 & 1 & 1\\
    0 & 0 & 1 & 1 & 0 & 0
    \end{bmatrix}
    \text{ and }
    \bm{y}=[1,1,1,0,0,0].
    \end{align} 
    \end{remark}
    Given a CSS code, we now characterize and represent all possible diagonal gates that realize a target diagonal logical gate. 
    Consider a diagonal physical gate $U_Z$ that preserves a CSS$(X,\mathcal{C}_2;Z,\mathcal{C}_1^\perp,\bm{y})$ code, inducing a diagonal logical gate $U_Z^l$. The generator coefficients $A_{\bm{0},\bm{\gamma}}$ appear as coefficients in the Pauli expansions of $U_Z$ and $U_Z^L$, creating a bridge between physical and logical worlds. We can express the coefficients $A_{\bm{0},\bm{\gamma}}$ in terms of the diagonal entries $d_{\bm{u}}$ of $U_Z$, and we can express them in terms of the diagonal entries $e^{\imath \theta_{\bm{\alpha}}}$ of the logical gate $U_Z^L$. Theorem \ref{thm:Main_conn_phy_log} results from equating these two expressions. 
    \begin{theorem}\label{thm:Main_conn_phy_log}
         Given a CSS($X,\mathcal{C}_2;Z,\mathcal{C}_1^\perp, \bm{y}$) code, the diagonal physical gate $U_Z= \sum_{\bm{u}\in \F_2^n} d_{\bm{u}}\ket{\bm{u}}\bra{\bm{u}}$ induces the logical gate $U_Z^L = \sum_{\bm{\alpha}\in \F_2^k} e^{\imath \theta_{\bm{\alpha}}}\ket{\bm{\alpha}}\bra{\bm{\alpha}}$ if and only if  
         \begin{equation}\label{eqn:con_phy_log}
             d_{\bm{u}\oplus\bm{y}} = e^{\imath\theta_{\bm{\alpha}}}~~~~ \text{  for  } G_{\mathcal{C}_2^\perp/\mathcal{C}_1^\perp} \bm{u}^T = \bm{\alpha}^T.
         \end{equation}
    \end{theorem}
    \begin{remark}
      If we think of $G_{\mathcal{C}_2^\perp/\mathcal{C}_1^\perp}\bm{v}^T$ as a syndrome, then we can observe that $\bm{u}$ and $\bm{u}+\bm{w}$, $\bm{w}\in\mathcal{C}_2$ determine the same syndrome. 
    \end{remark}
    \begin{proof}
    We express the generator coefficients $A_{\bm{0},\bm{\gamma}}$ in terms of the diagonal entries $e^{\imath\theta_{\bm{\alpha}}}$ of the logical gate $U_Z^L$,
    \begin{align} \label{eqn:log_Pauli_coeff1}
        \left[A_{\bm{0},\bm{\beta}G_{\mathcal{C}_2^\perp/\mathcal{C}_1^\perp}}\right]_{\bm{\beta}\in\F_2^k} =  
        \left[e^{\imath \theta_{\bm{\alpha}}}\right]_{\bm{\alpha}\in\F_2^k} \frac{1}{2^k}\left[(-1)^{\bm{\alpha\beta^T}}\right]_{\bm{\alpha},\bm{\beta}\in\F_2^k}.
    \end{align}
    We then express the coefficients $A_{\bm{0},\bm{\gamma}}$ in terms of the diagonal entries $d_{\bm{u}}$ of $U_Z$,
    \begin{align}
        \left[A_{\bm{\mu},\bm{\beta}G_{\mathcal{C}_2^\perp/\mathcal{C}_1^\perp}}\right]_{\bm{\beta}\in\F_2^k}
	   & = 
	   \frac{1}{|C_1|} \left[d_{\bm{u}\oplus\bm{y}}\right]_{\bm{u}\in\mathcal{C}_1} H^{\bm{\mu}=\bm{0}}_{(\mathcal{C}_1,\mathcal{C}_2^\perp/\mathcal{C}_1^\perp)}\nonumber \\
	   & = 
	   \frac{1}{|C_1|} \left[d_{\bm{u}\oplus\bm{y}}\right]_{\bm{u}\in\mathcal{C}_1}
	   \left[(-1)^{\bm{\beta}G_{\mathcal{C}_2^\perp/\mathcal{C}_1^\perp}\bm{u}^T}\right]_{{\bm{u}\in\mathcal{C}_1},\bm{\beta}\in \F_2^k}.
	   \label{eqn:log_Pauli_coeff2}
    \end{align}
    We permute entries in $[d_{\bm{u}\oplus\bm{y}}]_{\bm{u}\in\mathcal{C}_1}$ and rows in $H^{\bm{\mu}=\bm{0}}_{(\mathcal{C}_1,\mathcal{C}_2^\perp/\mathcal{C}_1^\perp)}$ to group together elements from the same coset of $\mathcal{C}_2$ in $\mathcal{C}_1$. Given $\bm{u_1}, \bm{u_2} \in \mathcal{C}_2$ and $\bm{w_1}, \bm{w_2} \in \mathcal{C}_1$, we have 
    \begin{align}
        \sum_{\bm{\beta}\in\F_2^k} (-1)^{\bm{\beta}G_{\mathcal{C}_2^\perp/\mathcal{C}_1^\perp} \left(\bm{u_1}\oplus\bm{w_1}\oplus\bm{u_2}\oplus\bm{w_2}\right)^T}
        =
        \left\{ 
	\begin{array}{lc}	
    2^k,  &\text{ if } \bm{w_1} \oplus \bm{w_2} \in \mathcal{C}_2, \\
	0, &\text{ otherwise. }
	\end{array}\right. 
    \end{align}
    Hence
    \begin{align}
        H^{\bm{\mu}=\bm{0}}_{(\mathcal{C}_1,\mathcal{C}_2^\perp/\mathcal{C}_1^\perp)} 
        \left(H^{\bm{\mu}=\bm{0}}_{(\mathcal{C}_1,\mathcal{C}_2^\perp/\mathcal{C}_1^\perp)}\right)^T
        = I_{2^{k_1-k_2}} \otimes B,
    \end{align}
    where $B$ is a square matrix of size $2^{k_2}$ with every entry equal to $2^k$. 
    We multiply \eqref{eqn:log_Pauli_coeff1} on the right by $\left(H^{\bm{\mu}=\bm{0}}_{(\mathcal{C}_1,\mathcal{C}_2^\perp/\mathcal{C}_1^\perp)}\right)^T$ to obtain
    \begin{align} \label{eqn:log_Pauli_coeff3}
        \left[ \frac{1}{2^{k_1}} \sum_{\bm{\alpha}\in\F_2^k} e^{\imath \theta_{\bm{\alpha}}} \sum_{\bm{\beta}\in\F_2^k} (-1)^{\bm{\beta} \left( \bm{\alpha}^T+G_{\mathcal{C}_2^\perp/\mathcal{C}_1^\perp}\bm{u}^T \right)}\right] = \left[e^{\imath \theta_{\bm{\alpha}(\bm{u})}}\right]_{\bm{\alpha}(\bm{u})^T = G_{\mathcal{C}_2^\perp/\mathcal{C}_1^\perp}\bm{u}^T}.
    \end{align}
    We then multiply \eqref{eqn:log_Pauli_coeff2} on the right by $\left(H^{\bm{\mu}=\bm{0}}_{(\mathcal{C}_1,\mathcal{C}_2^\perp/\mathcal{C}_1^\perp)}\right)^T$ to obtain
    \begin{align}\label{eqn:log_Pauli_coeff4}
        \frac{1}{2^{k_1}}\left[ d_{\bm{u}\oplus\bm{y}}\right]_{\bm{u}\in\mathcal{C}_1} \left(I_{2^{k_1-k_2}}\otimes B\right)
        =
        \left[ \frac{1}{2^{k_2}} \sum_{\bm{u}\in\mathcal{C}_2+\bm{w}}d_{\bm{u}\oplus\bm{y}} \right]_{\bm{w}\in \mathcal{C}_1/\mathcal{C}_2}
        = 
        \left[d_{\bm{u}\oplus\bm{y}} \right]_{\bm{\alpha}(\bm{u})^T = G_{\mathcal{C}_2^\perp/\mathcal{C}_1^\perp}\bm{u}^T}.
    \end{align}
    We conclude the proof by equating \eqref{eqn:log_Pauli_coeff3} and \eqref{eqn:log_Pauli_coeff4}.
    \end{proof}

    \begin{corollary}\label{coro:whether_invar}
    Set $
    \mathcal{C}_2 = \{\bm{u_0},\bm{u_1},\dots,\bm{u_{2^{k_2}-1}}\}$. A diagonal physical gate $U_Z= \sum_{\bm{u}\in \F_2^n} d_{\bm{u}}\ket{\bm{u}}\bra{\bm{u}}$ preserves the CSS($X,\mathcal{C}_2;$ $Z,\mathcal{C}_1^\perp$,$\bm{y}$) codespace if and only if for each fixed $\bm{w}\in \mathcal{C}_1/\mathcal{C}_2$, 
      $  d_{\bm{u_0}\oplus\bm{w}\oplus\bm{y}} =  d_{\bm{u_1}\oplus\bm{w}\oplus\bm{y}} = \dots = d_{\bm{u_{2^{k_2}-1}}\oplus\bm{w}\oplus\bm{y}}.$
    The induced logical operator is $U_Z^L = \sum_{\bm{\alpha}\in\F_2^k} d_{\bm{u_0}\oplus\bm{\alpha}G_{\mathcal{C}_1/\mathcal{C}_2}\oplus\bm{y}} \ket{\bm{\alpha}}\bra{\bm{\alpha}}$.
    \end{corollary}
    \begin{proof}
      Note that $G_{\mathcal{C}_1/\mathcal{C}_2} G^T_{\mathcal{C}_2^\perp/\mathcal{C}_1^\perp} = I_{k}$. Follows Theorem \ref{thm:general_invar} and Theorem \ref{thm:Main_conn_phy_log}. See Fig. \ref{fig:main_coro} for visualization.
    \end{proof}
    \begin{figure}[h!]
        \centering
        \begin{tikzpicture}
     \node (d1) at (-0.5,2.5) {$2^n$};
     \path[draw,thick] (0.2,0) -- (0,0) -- (0,5) -- (0.2,5); 
     \path[draw,thick] (4.8,0) -- (5,0) -- (5,5) -- (4.8,5); 
     \path[draw,dotted,thick] (3,-0.35) -- (3,5.35);
     \path[draw,dotted,thick] (-0.35,2) -- (5.35,2);
     \node (C1) at (1.5,5.5) {$\mathcal{C}_1$};
     \node (outside_C1) at (4,5.5) {$\F_2^n \setminus \mathcal{C}_1$};
     \path[draw,dotted] (0,4.1) -- (0.9,4.1) -- (0.9,5);
     \path[draw,dotted] (0.9,3.2) -- (0.9,4.1) -- (1.8,4.1) -- (1.8,3.2) -- (0.9,3.2);
     \path[draw,dotted] (2.1,2) -- (2.1,2.9) -- (3,2.9);
     \node (C2) at (0.45,3.9) {\scriptsize $\mathcal{C}_2$};
     \node (C2_w1) at (1.35,3) {\scriptsize $\mathcal{C}_2+\bm{w}_1$};
     \node (C2_w2) at (2.35,1.8) {\tiny $\mathcal{C}_2+\bm{w}_{2^k-1}$};
     \path[draw,thick] (9.2,2) -- (9,2) -- (9,4) -- (9.2,4); 
     \path[draw,thick] (10.8,2) -- (11,2) -- (11,4) -- (10.8,4); 
     \node (d2) at (8.65,3) {$2^k$};
     \node (d3) at (10,1.75) {$2^k$};
     \node (CSS) at (7,2.75) {an $[[n,k,d]]$};
     \node (CSS) at (7,2.25) {CSS code};
     \node (from) at (5.5,2.75) {};
     \node (to) at (8.45,2.75) {};
     \draw[<->,blue] (from) to [out=80,in=130] (to);
     \path[draw,ultra thick,red!60] (0.05,4.95) -- (3,2);
     \path[draw,ultra thick,red!10] (3,2) -- (4.95,0.05);
     \fill[blue!60](0.1,4.9)circle(1mm);
     \fill[blue!60](1,4)circle(1mm);
     \fill[blue!60](1.95,3.05)circle(1mm);
     \fill[blue!60](2.2,2.8)circle(1mm);
     
      \fill[blue!60](9.2,3.8)circle(1mm);
     \fill[blue!60](9.4,3.6)circle(1mm);
     \fill[blue!60](9.6,3.4)circle(1mm);
     \fill[blue!60](10.8,2.2)circle(1mm);
     \draw[loosely dotted] (9.8,3.2)-- (10.6,2.4);
    \end{tikzpicture}
        \caption{Bridge between Physical gate (left) and Induced Logical gate (right): If the little diagonal blocks of physical unitary are $aI_{2^{k_1-k_2}}$ for some constant $a\in \C$, then the physical gate preserves the CSS codespace and inducing the logical gate on the right by shrinking each little diagonal block into one diagonal element. }
        \label{fig:main_coro}
    \end{figure}
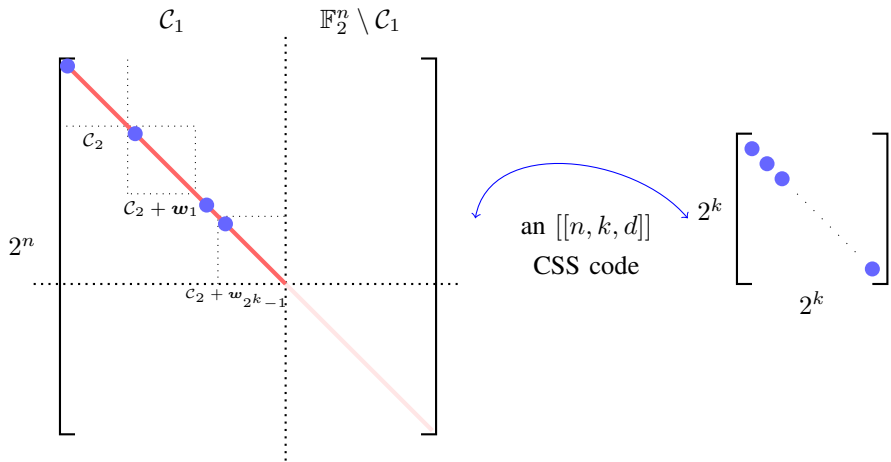
    \begin{remark}
      Corollary \ref{coro:whether_invar} provides a direct way to check whether a physical gate preserves a CSS code, and enables design of CSS codes that are preserved by a particular physical diagonal gate. It also implies that a CSS code with more $Z$-stabilizers (smaller $|\mathcal{C}_1|$) can be preserved by more physical diagonal gates, which is consistent with \cite[Theorem 2]{rengaswamy2020optimality}. When $U_Z=\sum_{\bm{u}\in\F_2^n} \left(e^{\imath\theta} \right)^{w_H(\bm{u})}\ket{\bm{u}}\bra{\bm{u}}$ and $\bm{y}=\bm{0}$, Corollary \ref{coro:whether_invar} can be interpreted as \cite[Corollary 3]{zeng2011transversality}. Here, we consider more general transversal physical gates (See Example \ref{exam:5_1_2}) and specify the induced logical gate explicitly. 
    \end{remark}
    
    In Example \ref{exam:15_1_3}, we see from Table \ref{tab:15_1_3_C_1} that all the weight-0 and weight-8 vectors are in $\mathcal{C}_2$ while all the weight-7 and weight-15 vectors are in the coset $\mathcal{C}_2+\bm{1}$. Thus diagonal entries in the same coset of $\mathcal{C}_2$ in $\mathcal{C}_1$ are identical. We now discuss the $\llbr 5,1,2\rrbr$ code \cite{vasmer2021morphing} introduced by Vasmer and Kubica, where the mixed transversal physical gate $P\otimes P^\dagger\otimes P \otimes \mathrm{C}Z$ induces a fault-tolerant logical $P$ gate.
    \begin{example}\label{exam:5_1_2}
    We first revisit the construction of the $\llbr 5,1,2\rrbr$ code \cite{vasmer2021morphing} starting from the stabilizer generator matrix (all positive signs $\bm{r}=\bm{y}=\bm{0}$).
    \begin{align}
    \setlength\aboverulesep{0pt}\setlength\belowrulesep{0pt}
    \setlength\cmidrulewidth{0.5pt}
    G_S = 
    \left[
    \begin{array}{ccccc|ccccc}
    1 & 1 & 0 & 1 & 0 & 0 & 0 & 0 & 0 & 0  \\
    0 & 1 & 1 & 0 & 1 & 0 & 0 & 0 & 0 & 0  \\
    \hline 
    0 & 0 & 0 & 0 & 0 & 1 & 1 & 0 & 0 & 1  \\
    0 & 0 & 0 & 0 & 0 & 0 & 1 & 1 & 1 & 0  \\
    \end{array}
    \right].
    \end{align}
    The only non-trivial $X$-logical is $\bm{w}=[1,1,1,0,0]\in \mathcal{C}_1/\mathcal{C}_2$. We have $\mathcal{C}_2 = \{\bm{0},[1,1,0,1,0],[0,1,1,0,1],[1,0,1,1,1]\}$ and $\mathcal{C}_2 +\bm{w}= \{\bm{w},[0,0,1,1,0],[1,0,0,0,1],[0,1,0,1,1]\}$. Consider the physical diagonal gate $U_Z = P\otimes P^\dagger\otimes P \otimes \mathrm{C}Z = \sum_{\bm{u}\in \F_2^5} d_{\bm{u}}\ket{\bm{u}}\bra{\bm{u}}$, we have 
    \begin{align}
        1&=d_{\bm{0}} 
          =e^{\imath\frac{\pi}{2}}e^{-\imath\frac{\pi}{2}} = d_{[1,1,0,1,0]} 
          =e^{-\imath\frac{\pi}{2}}e^{\imath\frac{\pi}{2}} = d_{[0,1,1,0,1]} 
          = e^{\imath\frac{\pi}{2}}e^{\imath\frac{\pi}{2}}e^{\imath \pi} = d_{[1,0,1,1,1]},\\
        e^{\imath\frac{\pi}{2}}&=e^{\imath\frac{\pi}{2}}e^{-\imath\frac{\pi}{2}}e^{\imath\frac{\pi}{2}} = d_{\bm{w}}
         = d_{[0,0,1,1,0]} = d_{[1,0,0,0,1]} = e^{-\imath\frac{\pi}{2}}e^{\imath \pi} = d_{[0,1,0,1,1]}.
    \end{align}
    It follows from Corollary \ref{coro:whether_invar} that $U_Z$ preserves the codespace, inducing the logical Phase gate $U_Z^L = \ket{0}\bra{0} + e^{\imath\frac{\pi}{2}}\ket{1}\bra{1}$. To demonstrate fault-tolerance, we first calculate the set of undetectable $Z$-errors,
    \begin{align}
        U_{e} = \{[1,1,1,0,0],~[0,0,1,1,0],~[1,0,0,0,1],~[0,1,0,1,1]\}.
    \end{align}
    Since the only two weight-$2$ undetectable errors are not confined to the support of $2$-local physical gate C$Z$, the logical Phase gate is fault-tolerant.
     
     
    \end{example}
    Theorem \ref{thm:Main_conn_phy_log} and Corollary \ref{coro:whether_invar} can be extended to general non-CSS stabilizer codes by the results in \cite[Appendix]{hu2021designing}. We consider a general stabilizer code generated by the matrix $	G_{\mathcal{S}}=\left[\begin{array}{c c}
		A & 0 \\ \hline
		0 & B\\ \hline 
		C & D\\
	\end{array} \right],$ where the submatrices $A$ and $B$ are maximized. Then, the results keep the same except to switch the tower of classical codes from $\mathcal{C}_2 \subset \mathcal{C}_1$ to $\langle A,C\rangle \subset B^\perp$. We illustrate the generalized Corollary \ref{coro:whether_invar} using the $\llbr 5,1,3 \rrbr$ stabilizer code to target a a logical $T$ gate. 
	
    \begin{example}\label{exam:5_1_3}
    Consider the $\llbr 5,1,3\rrbr$ stabilizer code with generator matrix $G_{\mathcal{S}}=[C|D]$, where
		\begin{align}
		C = 
		\left[
		\begin{array}{ccccc}
		1 & 0 & 0 & 1 & 0 \\
		0 & 1 & 0 & 0 & 1 \\
		1 & 0 & 1 & 0 & 0 \\
		0 & 1 & 0 & 1 & 0 \\
		\end{array}
		\right]
		\text{ and }
		D= 
		\left[
		\begin{array}{ccccc}
		0 & 1 & 1 & 0 & 0  \\
		0 & 0 & 1 & 1 & 0  \\
		0 & 0 & 0 & 1 & 1  \\
		1 & 0 & 0 & 0 & 1  \\
		\end{array}
		\right].
		\end{align}
		Note that $B=\{\bm{0}\}$, so $B^\perp = \F_2^5$. Consider the coset $\langle C \rangle$ in $\F_2^5$, where $\langle C \rangle$ contains all the even-weight vectors while its non-trivial coset includes all the odd-weight vectors. Then, it follows from Corollary \ref{coro:whether_invar} with the replaced tower $\langle C \rangle \subset \F_2^5$ that the only diagonal physical gate that preserves the $\llbr 5,1,3\rrbr$ code space and induces a logical $T$ gate is 
		\begin{align}
        U_Z
        &= \sum_{\bm{\alpha}\in \F_2^{5}} d_{\bm{\alpha}} \ket{\bm{\alpha}}\bra{\bm{\alpha}}, \text{ where } 
        d_{\bm{\alpha}} = 
         \left\{\begin{array}{lc}
         1, & \text{ if $w_H(\bm{\alpha})$ is even,} \\
         e^{\imath\frac{\pi}{4}}, & \text{ if $w_H(\bm{\alpha})$ is odd,} 
        \end{array} \right. \label{eqn:max_ent_UZ}\\
        &\equiv \exp\left(-\imath\frac{\pi}{8}Z\otimes Z\otimes Z\otimes Z \otimes Z\right).
        \end{align}
         Although $U_Z$ is a $5$-local gate, we can design a outer code that supports a fault-tolerant logical $U_Z$ (see Section \ref{sec:new_construction}).
    \end{example}
   The generator coefficient framework can work either forwards form a general diagonal physical gate as Example \ref{exam:5_1_2} or backwards from a target diagonal logical gate as Example \ref{exam:5_1_3}. In the following Section, we use the divisibility conditions of cosets in classical coding theory to construct a new family of CSS codes that is preserved by the transversal physical $T^\dagger$ gate, inducing a target logical gate. 
    
	\section{A New Family of CSS Codes Defined by Quadratic Forms} \label{sec:new_construction}
	The defining property of a classical divisible code \cite{Ward} is that codeword weights share a common divisor larger than one. Divisible codes can enable resilient quantum computation when a diagonal physical $U_Z$ preserves a CSS($X,\mathcal{C}_2; Z, \mathcal{C}_1^\perp,\bm{y}$) code, inducing a target logical operator $U_Z^L$. In Example \ref{exam:5_1_3}, the requirement that the induced logical operator is a $T$-gate forces half the cosets in $\mathcal{C}_1/\mathcal{C}_2$ to have all weights divisible by $8$, and the other half to have all weights congruent to $7$  modulo $8$. 
	
	Classical Reed-Muller codes are prototypical divisible codes. Codewords in the Reed-Muller codes RM$(r,m)$ are evaluation functions $[h(\bm{x})]_{\bm{x}\in \F_2^m}$ of boolean functions $h\in \F_2[x_1,\dots,x_m]$ of degree at most $r$. It follows from the theorem of Ax on polynomial zeros \cite{Ax}, that all weights in RM$(r,m)$ are divisible by $2^{\left\lfloor(m-1)/r\right\rfloor}$ (see also \cite{McEliece,macwilliams1977theory}). 
	
	Codewords in the first order Reed-Muller code RM$(1,m)$ are evaluation functions $[\epsilon\bm{1}\oplus L_{\bm{a}}(\bm{x})]_{\bm{x}\in\F_2^m}$ where $\epsilon\in\{0,1\}$ and $L_{\bm{a}} = a_1x_1\oplus \cdots a_m x_m$ is the linear function determined by a non-zero vector $\bm{a}\in\F_2^m$. RM$(1,m)$ is a $[2^m,m+1]$ code, and if we puncture on the coordinate $\bm{x}=\bm{0}$, we obtain the $[2^m-1,m]$ simplex code $\mathcal{C}(m)$, with all non-zero weights equal to $2^{m-1}$.
	
	Codewords in the second order Reed-Muller code RM$(2,m)$ are evaluation functions $[\epsilon\bm{1}\oplus L_{\bm{a}}(\bm{x})\oplus Q_{R}(\bm{}x)]_{\bm{x}\in\F_2^m}$ where $\epsilon\in\{0,1\}$, $L_{\bm{a}} = \bm{a}\bm{x}^T$ is a linear function and $Q_R(x)$ is a quadratic form. The property that defines a quadratic form is
	\begin{align}
	    Q_{R}(\bm{x}\oplus\bm{y}) = Q_R(\bm{x}) \oplus Q_R(\bm{y}) \oplus \bm{x}R\bm{y}^T,
	\end{align}
	where $R$ is a binary symmetric matrix with zero diagonal (binary symplectic matrix). Note that if we write $R=U+U^T$, where $U$ is strictly upper triangular, then we may set $Q_R(\bm{x})=\bm{x} U \bm{x}^T$. Observe that if $L_{\bm{a}}(\bm{x})$ is a linear function, then $Q_R(\bm{x})+L_{\bm{a}}(\bm{x})$ is a quadratic form corresponding to the same binary symplectic matrix $R$. 
	
	The weight distribution of the coset RM$(1,m)+[Q_R(\bm{x})]_{\bm{x}\in\F_2^m}$ depends only on the rank of the binary symplectic matrix $R$ (see for \cite{macwilliams1977theory} a proof using Dickson normal form). Lemma \ref{lemma:modulo_first_order} provides an alternative derivation based on the observation that $Q_R(x)$ is linear on the null space of $R$.

	\begin{lemma}\label{lemma:modulo_first_order}
	$wt_H\left(\left[ L_{\bm{a}(\bm{x}) +Q_R(\bm{x})} \right]_{\bm{x}\in\F_2^m}\right) = 2^{m-1}$ or $2^{m-1}\pm 2^{m-h-1}$.
    \begin{enumerate}
    \item All weights in the coset $\mathcal{C}(m)+[Q_R(\bm{x})]_{\bm{x}\in\F_2^m,\bm{x}\neq\bm{0}}$ are divisible by $2^{m-h-1}$.
    \item All weights in the coset $\mathcal{C}(m)+[Q_R(\bm{x})]_{\bm{x}\in\F_2^m,\bm{x}\neq \bm{0}}+\bm{1}$ are congruent to $2^{m-h-1}-1$ modulo ${2^{m-h-1}}$.
    \end{enumerate}
	\end{lemma}
	\begin{proof}
	We calculate the weight distribution $wt_H\left(\left[ L_{\bm{a}(\bm{x}) +Q_R(\bm{x})} \right]_{\bm{x}\in\F_2^m}\right)$ as $L_{\bm{a}}(\bm{x})$ ranges over the space of linear functions. Note that Rank$(R)=2h$ is even. Observe that the restriction of $Q_R(\bm{x})$ to the $(m-2h)$-dimensional space $V_R=\{\bm{x}\in\F_2^m |\bm{x}R=\bm{0}\}$ is a linear map. Hence
	\begin{align}
	     S_{\bm{a}} \coloneqq \sum_{\bm{x}\in V_R} (-1)^{Q_R(\bm{x})+L_{\bm{a}}(\bm{x})} 
	      = \left\{\begin{array}{lc}
         2^{m-2h}, & \text{ if $L_{\bm{a}}(\bm{x}) = Q_R(\bm{x})$ for all $\bm{x}\in V_R$,} \\
         0, & \text{ otherwise.} 
        \end{array} \right. 
	\end{align}
	Let $w_{\bm{a}}=wt_H(\left[ L_{\bm{a}(\bm{x})} +Q_R(\bm{x}) \right]_{\bm{x}\in\F_2^m})$. Then
	\begin{align}
	    2^m-2w_{\bm{a}} \eqqcolon T_{\bm{a}} = \sum_{\bm{x}\in \F_2^m} (-1)^{Q_R(\bm{x})+L_{\bm{a}}(\bm{x})}.
	\end{align}
	We square $T_{\bm{a}}$ to obtain 
	\begin{align}
	    T_{\bm{a}}^2 
	    &= \sum_{\bm{x}\in \F_2^m}\sum_{\bm{y}\in \F_2^m} (-1)^{Q_R(\bm{x})+Q_R(\bm{y})+L_{\bm{a}}(\bm{x})+L_{\bm{a}}(\bm{y})} \\
	    &=\sum_{\bm{x}\in \F_2^m}\sum_{\bm{y}\in \F_2^m} (-1)^{Q_R(\bm{x}\oplus\bm{y})+L_{\bm{a}}(\bm{x}\oplus\bm{y})+\bm{x}R\bm{y}^T}.
	\end{align}
	We change variables and sum over $\bm{z}=\bm{x}\oplus\bm{y}$ and $\bm{y}$. Note that $\bm{x}R\bm{x}^T=0$ for all $\bm{x}\in\F_2^m$ since $R$ has zero diagonal. Then 
	\begin{align}
	    T_{\bm{a}}^2 = \sum_{\bm{z}\in \F_2^m} (-1)^{Q_R(\bm{z})+L_{\bm{a}}(\bm{z})}\sum_{\bm{y}\in\F_2^m}(-1)^{(\bm{y}\oplus\bm{z})R\bm{y}^T}= 2^m S_{\bm{a}}.
	\end{align}
	Hence $T_{\bm{a}}=0$ or $T_{\bm{a}} = \pm 2^{m-h}$, and $w_{\bm{a}} \in \{2^{m-1},2^{m-1}\pm 2^{m-h-1}\}$. Parts $1)$ and $2)$ follow from the observation that $L_{\bm{a}}(\bm{0}) = Q_R(\bm{0}) = 0$, so puncturing on the zero coordinate does not change the weight. 
	\end{proof}
	Suppose $m \ge 4$. Consider the CSS$(X,\mathcal{C}_2;Z,\mathcal{C}_1^\perp,\bm{y}=\bm{0})$ code, where $\mathcal{C}_2 = \mathcal{C}(m)$ is the simplex code of length $n=2^m-1$ and 
	\begin{align}
	    \mathcal{C}_1 = \langle \mathcal{C}_2, [\bm{1}\oplus x_ix_j]_{\bm{x}\in\F_2^m,\bm{x} \neq \bm{0}} \mid 1\le i \le m-4, i<j \rangle.
	\end{align}
	The matrix 
	\begin{align}
	    G_{\mathcal{C}_1/\mathcal{C}_2} = 
	        \left[ \begin{array}{c}
	        \bm{1} \\ 
	        (\bm{1}\oplus x_ix_j)_{\bm{x}\in\F_2^m,\bm{x} \neq \bm{0}}\\ 
	        \dots\\ 
	        \end{array}\right]_{1\le i \le m-4,i<j}
	\end{align}
	generates the $X$-logicals. The minimum distance $d$ is the minimum distance of the Hamming code $\mathcal{C}_2^\perp$, so the parameters of the CSS code are $\llbr n,k= 1+\sum_{i=1}^{m-4}(m-i),d=3 \rrbr$.
	\begin{theorem}\label{thm:new_construction}
	The transversal $T^\dagger$ gate $U_Z = (T^\dagger)^{\otimes n} = \sum_{\bm{u}\in\F_2^n} \left(e^{\imath\frac{\pi}{4}} \right)^{w_H(\bm{u})}\ket{\bm{u}}\bra{\bm{u}}$ preserves the CSS$(X,\mathcal{C}_2;Z,\mathcal{C}_`^\perp,\bm{y}=\bm{0})$ code, inducing the logical operator 
	   \begin{align}
	    U_Z^L
        &= \sum_{\bm{\alpha}\in \F_2^{k}} d_{\bm{\alpha}} \ket{\bm{\alpha}}\bra{\bm{\alpha}}, \text{ where } 
        d_{\bm{\alpha}} = 
         \left\{\begin{array}{lc}
         1, & \text{ if $w_H(\bm{\alpha})$ is even,} \\
         e^{\imath\frac{\pi}{4}}, & \text{ if $w_H(\bm{\alpha})$ is odd,} 
        \end{array} \right. \label{eqn:max_ent_UZ_k}\\
        &\equiv \exp\left(\imath\frac{\pi}{8}Z\otimes Z \otimes \dots \otimes Z \right).\label{eqn:max_ent_UZ_k_exp}
	     \end{align}
	\end{theorem}
	\begin{remark} \label{rem:less_log}
	  Observe that \eqref{eqn:max_ent_UZ_k} and \eqref{eqn:max_ent_UZ_k_exp} only differ by a global phase $e^{-\imath\pi/8}$ and that \eqref{eqn:max_ent_UZ_k_exp} can be obtained from a single $T$ gate by conjugation, using a sequence of controlled-NOT gates. 
	\end{remark}
	\begin{proof}
	Given $\xi_{i,j}=$ $0$ or $1$ for $1\le i \le m-4$, $i<j$, we observe that the rank of the symplectic matrix $R$ determined by the quadratic form $Q_R(\bm{x})=\sum_{1\le i\le m-4, i<j} \xi_{i,j}x_ix_j$ is at most $2(m-4)$. Even weight $X$-logicals correspond to cosets $\mathcal{C}_2+[Q_R(\bm{x})]_{\bm{x}\in\F_2^m, \bm{x} \neq \bm{0}}$ and odd weight $X$-logicals correspond to cosets $\mathcal{C}_2+[Q_R(\bm{x})]_{\bm{x}\in\F_2^m, \bm{x} \neq \bm{0}} + \bm{1}$. Since $m-(m-4)-1=3$, it follows from Lemma \ref{lemma:modulo_first_order} that all weights in $\mathcal{C}_2+[Q_R(\bm{x})]_{\bm{x}\in\F_2^m, \bm{x} \neq \bm{0}}$ are congruent to $0$ modulo $8$, and that all weights in $\mathcal{C}_2+[Q_R(\bm{x})]_{\bm{x}\in\F_2^m, \bm{x} \neq \bm{0}} + \bm{1}$ are congruent to $7$ modulo $8$. It now follows from Corollary \ref{coro:whether_invar} that the physical transversal gate $U_Z = (T^\dagger)^{\otimes n}$ preserves the CSS code and that the induced logical gate $U_Z^L$ is given by \eqref{eqn:max_ent_UZ_k}. 
	\end{proof}
	
	\begin{remark}
	  The conclusions of Theorem \ref{thm:new_construction} hold for any $\llbr n=2^m-1,1\le k\le 1+\sum_{i=1}^{m-4}(m-i),d=3 \rrbr$ CSS code obtained by deleting rows of the form $(\bm{1}\oplus x_ix_j)_{\bm{x}\in\F_2^m,\bm{x}\neq \bm{0}}$ from $G_{\mathcal{C}_1/\mathcal{C}_2}$.
	\end{remark}
	
    The next Lemma shows that the logical gate $U_Z^L$ given by \eqref{eqn:max_ent_UZ_k} can be decomposed into a $T$-gate on every logical qubit, Controlled-Phase$^\dagger$ on every pair of the logical qubits, and Controlled-Controlled-$Z$ on every triple of logical qubits. 
	\begin{lemma}
	   $\frac{\pi}{4} \binom{k}{1} - \frac{\pi}{2}\binom{k}{2} + \pi\binom{k}{3} =  
   \left\{\begin{array}{lc}
         {0} {\pmod {2\pi}}, & \text{ if $k\ge 1$ is {even}}, \\
         {\frac{\pi}{4}} {\pmod {2\pi}}, & \text{ if $k\ge 1$ is {odd}}.
    \end{array} \right.$
	\end{lemma}
	\begin{proof}
	When $k=1$, only the first term remains to $\frac{\pi}{4}$. When $k=2$, only the first two terms remain and they sum to $0$. For $k\ge 3$
	  \begin{align}
	      \frac{\pi}{4} \binom{k}{1} - \frac{\pi}{2}\binom{k}{2} + \pi\binom{k}{3} 
	      & = \pi\left( \frac{k}{4} - \frac{k(k-1)}{4} +\frac{k(k-1)(k-2)}{6} \right) \nonumber\\
	      & = \frac{\pi}{12}k(k-2)(2k-5)\nonumber \\
	      & =  \left\{\begin{array}{lc}
         \frac{\pi}{3}t(t-1)(4t-5)= {0} {\pmod {2\pi}}, & \text{ if $k=2t$ for $t\in\Z^{+}$}, \\
         \frac{\pi}{3}t(t-1)(4t+1)+\frac{\pi}{4} = {\frac{\pi}{4}} {\pmod {2\pi}}, & \text{ if $k=2t+1$  for $t\in\Z^{+}$}.
    \end{array} \right.
	  \end{align}
	  Given two integers $t,t-1$, one must be odd and one even. Given three integers $t,t-1,4t+1$ or $t,t-1,4t-5$ exactly one must be divisible by 3. This observation completes the proof.
	\end{proof}
	\begin{example}
	Setting $m=5$, we consider the $\llbr 31,5,3\rrbr$ CSS code preserved by $\left(T^\dagger\right)^{\otimes 31}$. Let 
    \begin{align}
    G_{\mathcal{C}_1} = 
    \left[ \begin{array}{c}
    G_{\mathcal{C}_1/\mathcal{C}_2} \\
    G_{\mathcal{C}_2}
    \end{array}\right],
    \text{ where }
    G_{\mathcal{C}_1/\mathcal{C}_2} = 
     \left[ \begin{array}{c}
     \bm{1} \\
     (\bm{1}\oplus x_1x_i)_{\bm{x}\in\F_2^5, \bm{x}\neq \bm{0}}
     \end{array}\right]_{i=2,\dots,5}
     \text{ and }
     G_{\mathcal{C}_2}
      \left[ \begin{array}{c}
     (x_i)_{\bm{x}\in\F_2^5, \bm{x}\neq \bm{0}}
     \end{array}\right]_{i=1,\dots,5}.
    \end{align}
    If $R_i, i = 2,\dots,5$ is the binary symmetric matrix determined by the quadratic form $x_1x_i$, then every matrix $R$ in $\langle R_i \mid i=2,\dots,5 \rangle$ has rank at most 2. Even weight $X$-logicals determine cosets $\mathcal{C}_2 +[Q_R(\bm{x})]_{\bm{x}\in\F_2^5, \bm{x}\neq \bm{0}}$ and odd weight $X$-logicals determine cosets $\mathcal{C}_2 +[Q_R(\bm{x})]_ + \bm{1}$. As $m-(m-4)-1=3$, Theorem \ref{thm:new_construction} implies $\left(T^\dagger\right)^{\otimes 31}$ preserves the CSS code, and that the induced logical operator is given by \eqref{eqn:max_ent_UZ}.
	\end{example}
	
	  We may obtain an $\llbr n,k,d \rrbr$ CSS code with $d>3$ that is preserved by the transversal $T$ gate, by switching the $X$-stabilizers from the simplex code to the dual of $2$-error-correcting BCH code, or to the punctured Reed-Muller code RM$^*(r,m)$ with higher degree $r\ge 2$. However, to maintain the congruence conditions, we need to increase the number of physical qubits. We may optimize the parameters $n,k,$ and $d$ of the CSS code by choosing different classical component codes. 
	
	\section{Designing Stabilizer Codes in Layers}\label{sec:doubly-encoded}
	In Section \ref{sec:GCs}, we start from a stabilizer code on $N_1$ qubits and derive all possible diagonal physical gates $U_Z'$ on $N_1$ qubits that induce a target logical gate. In Example 3, the unique (up to global phase) physical gate that preserves the $\llbr 5,1,3\rrbr$ code and induces a logical $T$ gate is specified by \eqref{eqn:max_ent_UZ}. 
	
	Section \ref{sec:new_construction} embeds the $N_1$ qubits in a larger physical space of $N_2$ qubits. The $N_1$ qubits become the logical qubits of a stabilizer code on $N_2$ qubits. The code is preserved by a transversal physical diagonal gate on $N_2$ qubits, inducing the operator $U_Z'$ on the $N_1$ code qubits. The transversal diagonal gate on $N_2$ qubits preserves the \emph{outer} code, inducing the target logical operator on the \emph{inner} code. For example, $\left(T^\dagger \right)^{\otimes 31}$ preserves the $5$ logical qubits of the $\llbr 31,5,3\rrbr $ code inducing a logical $T$gate on the inner $\llbr 5,1,3\rrbr$ code. 
	
	The same method applies to the $\llbr 7,1,3\rrbr$ Steane code, where the inner qubits on the $7$ logical qubits of a $\llbr 63,7,3\rrbr$ CSS code (constructed in the Remark \ref{rem:less_log}). Note however, that there are physical gates other than \eqref{eqn:max_ent_UZ_k} that induce a logical $T$ on the $\llbr 7,1,3\rrbr$ Steane code, so it may be possible to improve on the parameters of the outer code.
	
	Figure \ref{fig:two_layer} describes this method of designing stabilizer codes in two layers. What makes it feasible is the bridge between physical and logical quantum domains created by generator coefficients.  
	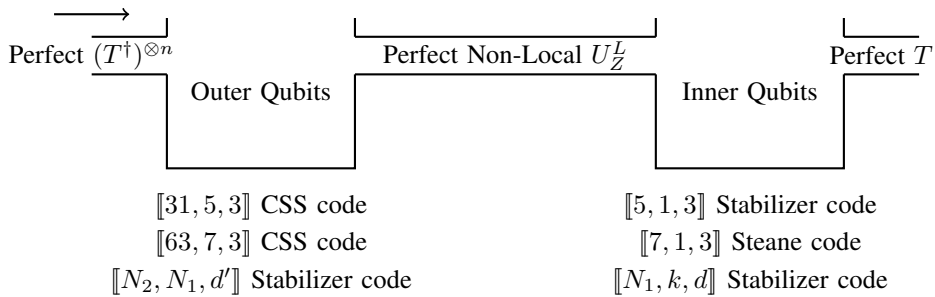
\begin{figure}[h!]
	    \centering
	       \begin{tikzpicture}
        \path[draw,thick] (0,0) -- (0,-0.25);
        \path[draw,thick] (0,-0.75) -- (0,-2) -- (2.5,-2) -- (2.5,-0.75); 
        \path[draw,thick] (0,-0.25) -- (-1,-0.25);
    	\path[draw,thick] (0,-0.75) -- (-1,-0.75);
    	\path[draw,thick] (2.5,-0.25) -- (6.5,-0.25);
    	\path[draw,thick] (2.5,-0.75) -- (6.5,-0.75);
    	\path[draw,thick] (2.5,-0.25) -- (2.5,0);
    	\node (d1) at (1.25,-1) {Outer Qubits};
    	\node (d2) at (1.25,-2.5) {$\llbr 31,5,3 \rrbr$ CSS code};
    	\node (d2') at (1.25,-3) {$\llbr 63,7,3 \rrbr$ CSS code};
        \node (d2'') at (1.25,-3.5) {$\llbr N_2,N_1,d' \rrbr$ Stabilizer code};
    	    	
    	\node (d) at (-1,-0.5) {Perfect $(T^\dagger)^{\otimes n}$};
    	\node (dd) at (4.5,-0.5) {Perfect Non-Local $U_Z^L$};
    	\path[draw,thick] (6.5,0) -- (6.5,-0.25);
        \path[draw,thick] (6.5,-0.75) -- (6.5,-2) -- (9,-2) -- (9,-0.75); 
        \path[draw,thick] (9,-0.25) -- (10,-0.25);
    	\path[draw,thick] (9,-0.75) -- (10,-0.75);
    	\path[draw,thick] (9,-0.25) -- (9,0);
    	\node (d3) at (7.75,-1) {Inner Qubits};   	
        \node (d4) at (9.5,-0.5) {Perfect $T$};
        \path[draw,thick,->] (-1.5,0.05) -- (-0.5,0.05);
        \node (d5) at (7.75,-2.5) {$\llbr 5,1,3 \rrbr$ Stabilizer code};
        \node (d5') at (7.75,-3) {$\llbr 7,1,3 \rrbr$ Steane code};
        \node (d5'') at (7.75,-3.5) {$\llbr N_1,k,d \rrbr$ Stabilizer code};
    \end{tikzpicture}
	    \caption{Configuring outer and inner qubits so that transversal $T^\dagger$ gate on outer qubits induces a logical $T$ gate on the inner qubit.}
	    \label{fig:two_layer}
	\end{figure}
	
	
	It is often possible to assemble a universal set of fault-tolerant gates on the inner qubits by bringing together transversal gates on both outer and inner qubits. It may be useful to view concatenation of the $\llbr 31,5,3\rrbr$ code and the $\llbr 5,1,3\rrbr$ code as factorization of a $\llbr 31,1,3\rrbr$ triorthogonal code. Factorization is similar to code switching \cite{anderson2014fault,paetznick2013universal,bombin2015gauge,hill2013fault} in that it can produce a universal set of gates without requiring teleportation of magic states. However, inducing a logical gate on the inner code does require encoding/decoding algorithms to pass between outer and inner codes. The overhead of factorization depends on the complexity of these algorithms. It is also interesting to optimize gate synthesis for a given quantum algorithm by minimizing the number of code switches.

	\section{Discussion}
	\label{sec:conclusion}
	Given a CSS($X$, $\mathcal{C}_2$; $Z$, $\mathcal{C}_1^\perp, \bm{y}$) code, we have used the mathematical framework of generator coefficients to characterize all physical diagonal gates that induce a target logical gate. When the logical gate is the identity, the physical gates represent types of noise to which the CSS code is oblivious (see \cite{hu2022mitigating}). Diagonal entries of the physical gate outside $\mathcal{C}_1+\bm{y}$ are unconstrained, and diagonal entries from the same coset of $(\mathcal{C}_2$ in $\mathcal{C}_1)+\bm{y}$ are required to be equal. Our framework is quite general, it includes CSS codes that are preserved by transversal $T$, also the hybrid codes of Vasmer and Kubica \cite{vasmer2021morphing}.
	
	The connection between physical and logical domains enables analysis of fault tolerance. We begin by observing that when a transversal $(d-1)$-local physical gate preserves an $\llbr n,k,d \rrbr$ CSS code, the induced logical operator is at least error-detectable. Our framework provides visibility into cases where a physical gate with higher locality still induces a fault-tolerant logical gate. Vasmer and Kubica \cite{vasmer2021morphing} have described examples where $2$-local physical gates induce a fault-tolerant Phase gate on the $\llbr 5,1,2\rrbr$ code, and where $3$-local physical gates induce a fault-tolerant $T$-gate on the $\llbr 10,1,2\rrbr$ code. Our framework shows fault-tolerance is preserved when the local support of a physical gate is not contained in the support of any undetectable error. It is easy to check this condition on the generator matrix $G_{\mathcal{C}_2}$ of $X$-stabilizers: the matrix obtained by puncturing $G_{\mathcal{C}_2}$ on the support of an $r$-local gate must be full rank.
	
	The constraints we derived on diagonal entries of a physical gate motivate the design of CSS codes where the component codes $\mathcal{C}_1$, $\mathcal{C}_2$ are classical divisible codes. We have introduced a new family of CSS codes defined by quadratic forms. The family motivates the design of stabilizer codes in layers, a way of producing a universal set of gates without requiring teleportation of magic states.

	\section*{Acknowledgement}
	This work was supported in part by the National Science Foundation (NSF) under Grant CCF-2106213 and Grant CCF-1908730. 
	
	\bibliographystyle{IEEEtran}
	\bibliography{bibliography}

	

\end{document}